\documentclass[draftcls, onecolumn, 12pt]{IEEEtran}
\usepackage{amsmath,amssymb,threeparttable,placeins, enumerate, caption}
\usepackage{mathtools,relsize}
\usepackage[pdftex]{graphicx}
\usepackage{epstopdf}
\usepackage{color}
\usepackage[usenames]{xcolor}
\usepackage{amsfonts}
\usepackage{latexsym}
\usepackage{amsthm}
\usepackage{subfigure, multirow,makecell,stfloats}

\usepackage{algorithm}
\usepackage{cite}
\usepackage[noend]{algpseudocode}
\usepackage{setspace}
\usepackage{tabularx}
\usepackage[figuresright]{rotating}

\def\qi#1 {\fbox {\footnote {\ }}\ \footnotetext { From Qi: {\color{red}#1}}}
\usepackage{amssymb}

\usepackage{amsthm}

\theoremstyle{remark} 
\newtheorem{theorem}{{{\textit{Theorem}}}}

\newtheorem{lemma}{{{\textit{Lemma}}}}
\newtheorem{corollary}{{{{\textit{Corollary}}}}}

\newtheorem{definition}{{{\textit{Definition}}}}

\newtheorem{remark}{{{\textit{Remark}}}}

\newtheorem{example}{{{\textit{Example}}}}

\newtheorem{construction}{{{\textit{Construction}}}}
\title{Asymptotically Optimal Aperiodic Doppler Resilient Complementary Sequence  Sets Via Generalized Quasi-Florentine Rectangles}
	\author{Zheng Wang, Zhiye Yang, Yang Yang,~\IEEEmembership{Member,~IEEE}, Avik Ranjan Adhikary,~\IEEEmembership{Member,~IEEE}, and Keqin Feng
	\thanks{ 
		Z. Wang, Y. Yang, and A. R. Adhikary are with the School of Mathematics, Southwest Jiaotong University, Chengdu, 611756, China. (e-mail: wang\_z@my.swjtu.edu.cn, yang\_data@swjtu.edu.cn, avik.adhikary@ieee.org).

		Zhiye Yang is with the Research Center for Number Theory and Its Applications School of Mathematics, Northwest University, Xi’an, 710127, Shaanxi, China. (e-mail: zyyang02@126.com).

		Keqin Feng is with the Department of Mathematical Sciences, Tsinghua University, Beijing, 100084, China. (e-mail:fengkq@tsinghua.edu.cn).
	}
}
\begin{document}
	\maketitle
	
	\begin{abstract}
		Doppler-resilient complementary sequence (DRCS) sets play a vital role in modern communication and sensing systems, particularly in high-mobility environments. This work makes two primary contributions. First, we refine the definition of quasi-Florentine rectangles to a more general form, termed generalized quasi-Florentine rectangles, and propose a systematic method for their construction. Second, we propose several sets of aperiodic DRCS based on generalized quasi-Florentine rectangles and Butson-type Hadamard matrices. The proposed aperiodic DRCS sets are shown to be asymptotically optimal with respect to the lower bound of aperiodic DRCS sets. 
	\end{abstract}
	
	
	\section{Introduction}
	Golay introduced the concept of Golay complementary pairs (GCPs) \cite {golay1961complementary}, characterized by the property that the sum of their aperiodic autocorrelation functions is zero for all non-zero time shifts. Tseng and Liu later generalized this foundational idea in \cite{tseng1972complementary}, where they extended GCPs to form complementary sequence sets (CSs) and further proposed mutually orthogonal complementary sequence sets (MOCSSs). These sets are designed so that the total aperiodic autocorrelation is zero at every non-zero time shift, while their cross-correlations also sum to zero for all shifts. CSs and MOCSSs have found widespread applications in various domains, including channel estimation \cite{channeles2001complementary}, spreading codes for asynchronous multicarrier code division multiple access (MC-CDMA) communications \cite{mccdma2001multicarrier}, reducing the peak to average power ratio (PAPR) in orthogonal frequency division multiplexing (OFDM) systems \cite{PAPR2021encoding}, Doppler resilient waveform design \cite{Doppler2008doppler}, \cite{Doppler2024complementary}, integrated sensing and communication (ISAC) \cite{ISAC2017ieee}, \cite{duggal2020doppler}.
	Despite their advantages, MOCSSs are constrained by a fundamental limitation: the size of MOCSSs cannot exceed the flock size, restricting their use in systems that require the support of a large number of users. To overcome this bottleneck, researchers have explored the use of quasi-complementary sequence sets (QCSSs) \cite{liu2013tighter}, in which the aggregate aperiodic correlation is allowed to be small but non-zero. 

	Modern communication systems are increasingly challenged by the Doppler effect induced by high-mobility environments. Traditional sequence design methodologies, which primarily optimize correlation functions, have proven inadequate for the demands of modern applications. In emerging paradigms such as ISAC systems, sequence design must account for both delay and Doppler. Sequences that exhibit robustness to Doppler shifts are known as Doppler-resilient sequences (DRSs) \cite{ye22}. Consequently, the focus of performance evaluation has transitioned from correlation functions to the ambiguity function (AF), which captures both delay and Doppler characteristics.
	Several researchers have explored the lower bounds of the maximum AF magnitude. In 2013, Ding \textit{et al.} \cite{ding13} established a lower bound on the maximum sidelobe level of the AF using an approach grounded in the Welch bound \cite{welch1974lower}. It is worth highlighting that in many practical systems (e.g., \cite{duggal2020doppler, kumari}), the ranges of interest for Doppler shifts and delays are often much narrower than the total signal bandwidth and the sequence length. Building on this observation,
	Ye \textit{et al.} \cite{ye22} introduced the concept of low/zero ambiguity zone (LAZ/ZAZ) and further derived lower bounds for LAZ/ZAZ sequence sets. More recently, Meng \textit{et al.} \cite{meng2024new} developed improved lower bounds for aperiodic AF under some specific parameter constraints. In parallel, several DRS constructions have been proposed that meet or approach these lower bounds of the AF \cite{ye22, Tian2025, wang2025asymptotically,wang25A,yang2025new}. However, designing the optimal set of aperiodic LAZ/ZAZ sequences on the entire delay-Doppler plane with flexible parameters remains a challenge.
	
	In order to reduce the magnitude of the AF, based on the idea of complementary sequences, Shen \textit{et al.} \cite{shen2024} proposed the concept of Doppler-resilient complementary sequence (DRCS). Their approach involves transmitting distinct sequences in each pulse and coherently accumulating the resulting AF. In addition, Shen \textit{et al.} \cite{shen2024} derived lower bounds for periodic, aperiodic, and odd periodic AF of DRCS sets and proposed several constructions, including optimal periodic DRCS sets by using circular Florentine rectangles. More recently, Wang \textit{et al.} \cite{wang2025asymptotically} established a tighter lower bound for the aperiodic AF of DRCS sets by the idea of the Levenshtein bound. Furthermore, the authors in \cite{wang2025asymptotically} proposed a construction of the asymptotically optimal DRCS sets based on the quasi-Florentine rectangles and the Butson-type Hadamard matrix. Quasi-Florentine rectangles are intriguing combinatorial objects that were first introduced by Adhikary \textit{et al.} in \cite{Avik2024}, and further, their circular properties were explored in \cite{adhikary2025periodic}. A matrix is called a quasi-Florentine rectangle if it satisfies all the properties of Florentine rectangles; however, it misses one element in each row.
	Compared to circular Florentine rectangles, quasi-Florentine rectangles possess a large number of rows when $N$ is a prime power or greater than a prime power. Indeed, using the quasi-Florentine rectangle, the Florentine rectangle, and the framework given in \cite{wang2025doppler}, we can generate DRCS sets for any given $N$. However, the size of the DRCS sets is equal to the size of the quasi-Florentine rectangles or Florentine rectangles in the construction framework given in \cite{wang2025doppler}. For some value of $N$, with a maximum number of rows of the Florentine rectangle or quasi-Florentine rectangle being $4$, it remains an open problem to construct the optimal aperiodic DRCS set.
	
	In this work, we provide a more refined definition of the quasi-Florentine rectangle, referred to as the generalized quasi-Florentine rectangle, allowing for the possibility of multiple missing elements in each row, while preserving the other properties of Florentine rectangles. We denote the number of rows in the quasi-Florentine rectangle as $ F_{Q,x}(N) $, where $ x $ represents the number of elements missing in each row. The main contributions of this paper are summarized as follows:
	\begin{itemize}
		\item We have presented a definition of generalized quasi-Florentine rectangles. We also propose some constructions of the generalized quasi-Florentine rectangles. The generalized quasi-Florentine rectangle significantly extends the range of admissible parameters and overcomes the limitations on the number of rows inherent in existing Florentine and quasi-Florentine constructions. In particular, for certain values of $N$ for which $F_{Q,1}(N) = F(N) = 4$, the generalized quasi-Florentine rectangle achieves a significantly larger number of rows than both the Florentine and quasi-Florentine rectangles. Representative examples are provided in Table~\ref{table_2}.
		\item Based on generalized quasi-Florentine rectangles and Butson-type Hadamard matrices, we construct several DRCS sets. Remarkably, the resulting DRCS sets are asymptotically optimal with respect to the lower bound of the aperiodic DRCS sets recently in \cite{wang2025doppler}. Moreover, for the case $N$ satisfying $F_{Q,1}(N) = F(N) = 4$, optimal DRCS sets are obtained. Examples are given in Tables~\ref{2} and \ref{3}. As a comparison with known constructions, the parameters of our proposed aperiodic DRCS sets are listed in Table \ref{Table1}. Obviously, the constructed DRCS sets have more flexible parameters in this paper.
	\end{itemize}
	
	\begin{table}[htp]
		\caption{Known Aperiodic DRCS Sets}
		\label{Table1}
		\renewcommand{\arraystretch}{1.3} 
		\resizebox{\textwidth}{!}{
			\begin{tabular}{|c|c|c|c|c|c|c|c|c|}
				\hline References & Set Size & \begin{tabular}{c} 
					Flock \\
					Size
				\end{tabular}  & \begin{tabular}{c} 
					Sequence \\
					Length
				\end{tabular}& $Z_x$ &$Z_y$&  $\hat{\theta}_{\text {max }}$ & Alphabet & Parameter Constraint(s) \\
				\hline \cite{shen2024} & $\widetilde{F}(N)$ & $N$ & $N$ & $N$&$N$& $N$ & $\mathbb{Z}_N$ & 
				\begin{tabular}{c}
					$N \geq 2$ is an integer, $\widetilde{F}(N)$ is the\\
					maximum number of rows for which an\\
					$\widetilde{F}(N) \times N$ circular
					Florentine rectangles exist. 
				\end{tabular}
				\\
				\hline \cite{wang2025doppler} & $F_{Q,1}(N)$ & $N$ & $N-1$ &$N-1$& $N-1$& $N$ & $\mathbb{Z}_r$ & 
				\begin{tabular}{c}
					$2\leq r \leq  N$ is an integer, $N \geq 2$\\
					is an integer, $F_{Q,1}(N)$ is the maximum \\
					number of rows for which an $F_{Q,1}(N) \times(N-1)$ \\
					quasi-Florentine rectangle exists. 
				\end{tabular} \\
				\hline \cite{wang2025doppler} & $F(N)$ & $N$ & $N$ &$N$&$N$& $N$ & $\mathbb{Z}_r$ & 
				\begin{tabular}{c}
					$2\leq r \leq  N$ is an integer, $N \geq 2$\\
					is an integer, $F(N)$ is the maximum \\
					number of rows for which an \\
					$F(N) \times N$  Florentine rectangle exists. 
				\end{tabular} \\
				\hline Theorem \ref{thc2} &$F_{Q,N-L}(N)$ & $N$ & $L$ &$L$&$L$& $N$ & $\mathbb{Z}_{r}$ & \begin{tabular}{c}
					$2\leq r \leq  N$ is an integer, $N \geq 2$\\
					is an integer, $F_{Q,N-L}(N)$ is the maximum \\
					number of rows for which an $F_{Q,N-L}(N) \times L$ \\
					quasi-Florentine rectangle exists, $2\leq L\leq  N$. 
				\end{tabular}\\
				\hline
			\end{tabular}
		}
	\end{table}
	
	The rest of the paper is organized as follows. Section II introduces the necessary notations and lemmas. In Section III, we have proposed the definced of generalized quasi-Florentine rectangles. We also give some systematic constructions of generalized quasi-Florentine rectangles in this section. In Section IV, we propose some new classes of DRCS sets with asymptotic optimality. In Section V, we make a comparison of our work with the existing works in the literature. Finally, Section VI concludes the paper.
	
	\section{Preliminaries}
	For convenience, we will use the following notations consistently in this paper:
	\begin{itemize}
		\item \( \mathbb{Z}_N \) denotes the ring of integers modulo \( N \).
		\item \( \omega_N = e^{\frac{2 \pi \sqrt{-1}}{N}} \) is a primitive \( N \)-th complex root of unity.
		\item \( (\cdot)^* \) denotes the complex conjugate.
		\item $\widetilde{F}(N)$ denotes the maximum number of rows for which
		an $\widetilde{F}(N)\times N$ circular Florentine rectangle exists through systematic construction in \cite{song1992aspects}.
		\item $F(N)$ denotes the maximum number of rows for which
		an $F(N)\times N$ Florentine rectangle exists through systematic construction in \cite{avik2021asymptotically}.
		\item $F_{Q,1}(N)$ denotes the maximum number of rows for which
		an $F_{Q,1}(N)\times (N-1)$ quasi-Florentine rectangle exists through systematic construction in \cite{Avik2024}.
		\item $F_{Q,N-n}(N)$ denotes the maximum number of rows for which
		an $F_{Q, N-n}(N)\times n$ generalized quasi-Florentine rectangle exists through systematic constructions discussed in this paper.
        \item $\widetilde{F}_{Q,N-n}(N)$ denotes the maximum number of rows for which
		an $\widetilde{F}_{Q, N-n}(N)\times n$ generalized circular quasi-Florentine rectangle exists through systematic constructions discussed in this paper.
	\end{itemize}

	\subsection{Ambiguity Function}
	Let $\mathbf{a}=(a(0),a(1),\cdots,a(N-1))$ and $\mathbf{b}=(b(0),b(1),$ $\cdots,b(N-1))$ be two complex unimodular sequences with period $N,$ i.e., $|a(i)|=1$ and $0\leq i<N$.
	The aperiodic cross-AF of $\mathbf{a}$ and $\mathbf{b}$ in time shift $\tau$ and Doppler $\nu$ is defined as follows:
	\begin{align*}
		\hat{AF}_{\mathbf{a}, \mathbf{b}}(\tau, \nu)=\left\{\begin{aligned}
			&\sum_{t=0}^{N-1-\tau} a(t) b^*(t+\tau) \omega_N^{\nu t}, & 0 \leq \tau \leq N-1,\\
			&\sum_{t=-\tau}^{N-1} a(t) b^*(t+\tau) \omega_N^{\nu t}, & 1-N \leq \tau<0,\\
			&0,&|\tau|\geq N.
		\end{aligned}\right.
	\end{align*}
	If $\mathbf{a}=\mathbf{b},$ then $\hat{AF}_{\mathbf{a}, \mathbf{b}}(\tau, \nu)$ is called aperiodic auto-AF and written as $\hat{AF}_{\mathbf{a}}(\tau,\nu)$.

	A DRCS set $\mathcal{C}=\left\{\mathbf{C}^{(0)}, \mathbf{C}^{(1)}, \cdots, \mathbf{C}^{(K-1)}\right\}$ contains $K$ DRCSs, each of which consists of $M\geq 2$ sequences of length $N$, i.e.,
	\begin{align*}
		\mathbf{C}^{(k)}=\left[\begin{array}{c}
			\mathbf{c}_0^{(k)} \\
			\mathbf{c}_1^{(k)} \\
			\vdots \\
			\mathbf{c}_{M-1}^{(k)}
		\end{array}\right]_{M \times N},
	\end{align*}
	where $\mathbf{c}_m^{(k)} =(c_{m}^{(k)}(0), c_{m}^{(k)}(1), \cdots, c_{m}^{(k)}(N-1))$, $0 \leq k < K$, and $0\leq m<M$. 
	For two DRCSs $\mathbf{C}^{(k_1)}$ and $\mathbf{C}^{(k_2)}$, their aperiodic cross-AF is defined as the aperiodic AF sum, i.e.,
	\begin{align*}
		\hat{AF}_{\mathbf{C}^{(k_1)}, \mathbf{C}^{(k_2)}}(\tau,\nu)=\sum_{m=0}^{M-1}\hat{AF}_{\mathbf{c}_m^{(k_1)}, \mathbf{c}_m^{(k_2)}}(\tau,\nu).
	\end{align*}
	If $k_1=k_2,$ then $\hat{AF}_{\mathbf{C}^{(k_1)}, \mathbf{C}^{(k_2)}}(\tau,\nu)$ is called aperiodic auto-AF and written as $\hat{AF}_{\mathbf{C}^{(k_1)}}(\tau,\nu)$.

	\subsection{DRCS Set and Aperiodic Bound}
	For a DRCS set $\mathcal{C}$, its maximum aperiodic AF magnitude over a region $\Pi =(-Z_x, Z_x) \times(-Z_y, Z_y)\subseteq(-N, N) \times(-N, N)$ is defined as $\hat{\theta}_{\max}(\mathcal{C})=\max \left\{\hat{\theta}_a(\mathcal{C}), \hat{\theta}_c(\mathcal{C})\right\}$, where
	\begin{align*}
		\hat{\theta}_a(\mathcal{C})=&\max \left\{\left|A F_{\mathbf{C}} (\tau, \nu)\right|: \mathbf{C}  \in \mathcal{C},(\tau, \nu) \neq(0,0) \in \Pi\right\},\\
		\hat{\theta}_c(\mathcal{C})=&\max \left\{\left|A F_{\mathbf{C}, \mathbf{D}}(\tau, \nu)\right|: \mathbf{C}, \mathbf{D} \in \mathcal{C},\mathbf{C}\neq  \mathbf{D}, (\tau, \nu) \in \Pi\right\}
	\end{align*}
	are called the maximum aperiodic auto-AF magnitude and the maximum aperiodic cross-AF magnitude, respectively.
	Then, $\mathcal{C}$ is referred to as an aperiodic $(K, M, N, \hat{\theta}_{\max},\Pi)$-DRCS set. 
	
	The aperiodic AF magnitude lower bound of DRCS sets can be expressed as follows:
	\begin{lemma}[\cite{wang2025doppler}]\label{lem-wangbound}
		For an aperiodic $\left(K, M, N, \hat{\theta}_{\max }, \Pi\right)$-DRCS set, where $\Pi=\left(-Z_x, Z_x\right) \times\left(-Z_y, Z_y\right), 1 \leq Z_x, Z_y \leq N$, the lower bound of the aperiodic AF magnitude is given by
		\begin{align}\label{C1eq}
			\hat{\theta}_{\max} \geq \sqrt{MN \left( 1 - 2 \sqrt{\frac{M}{3KZ_y}} \right)},
		\end{align}
		where \( K > \frac{3M}{Z_y} \) and \( N \sqrt{\frac{3M}{KZ_y}} \leq Z_x \leq N \). 
	\end{lemma}
	
	In the next step, to evaluate the closeness between the lower bounds and the achieved aperiodic AF, we define the optimality factor $\hat{\rho}$ as follows.
	
	\begin{definition}[Optimality Factor]
		For an aperiodic $\left(K, M, N, \hat{\theta}_{\max }, \Pi\right)$-DRCS set, where $\Pi = (-Z_x, Z_x) \times (-Z_y, Z_y),$ the optimality factor  $\hat{\rho}$ is as follows:
		\begin{align}\label{aperiodicoptimal1}
			\hat{\rho} = \frac{\hat{\theta}_{\max}}{\sqrt{MN \left( 1 - 2 \sqrt{\frac{M}{3KZ_y}} \right)}},
		\end{align}
		where \( K > \frac{3M}{Z_y} \) and \( N \sqrt{\frac{3M}{KZ_y}} \leq Z_x \leq N \). 
	\end{definition}
	The DRCS set is said to be optimal if $\hat{\rho}=1$, and is said to be asymptotically optimal if $\lim\limits_{N \rightarrow \infty}\hat{\rho} = 1$.

	\subsection{Butson-type Hadamard Matrix}
	\begin{definition}[\cite{had1973complex}]
		Let $N$ and $r$ be two positive integers, and let $\mathbf{B}=(\omega_r^{b_{i,j}})_{0\leq i,j<N}$ be a matrix of order $N$, where $b_{i,j} \in \mathbb{Z}$. If $\mathbf{B}\mathbf{B}^H=N\mathbf{I}$, then it is called a Butson-type Hadamard matrix, denoted by $BH(N,r)$, where $\mathbf{I}$ is the identity matrix of order $N$.
	\end{definition}
	
	\begin{remark}
		The discrete Fourier transform (DFT) matrices, the Walsh-Hadamard matrices, and the Hadamard matrices are special cases of Butson-type Hadamard matrices with parameters $BH(N, N)$, $BH(2^m,2)$, and $BH(4N,2) $, respectively. 
	\end{remark}
	
	\begin{lemma}[\cite{had1973complex}]\label{BH_lemma}
		Let $BH(L_1,r_1)$ and $BH(L_2,r_2)$ be two Butson-type Hadamard matrices. Using the Kronecker product, we can obtain $BH(L_1L_2,\text{lcm}\{r_1,r_2\})$, where $\text{lcm}$ denotes the least common multiple. 
	\end{lemma}
	This paper utilizes Butson-type Hadamard matrices to construct DRCSs with smaller alphabets. In \cite{hadc2006complex}, the authors presented Butson-type Hadamard matrices for various parameters, and Table \ref{BH} lists some known parameters of the seed Butson-type Hadamard matrices $BH(N,r)$ for \( r \leq 7 \). Furthermore, Wallis \cite{had1973complex} proposed various Butson-type Hadamard matrices of different dimensions and alphabet sizes utilizing the Kronecker product, as presented in Lemma \ref{BH_lemma}.

	\begin{table}[htp]
		\begin{center}
			\caption{Parameters of seed Butson-type Hadamarad matrices over alphabet size $\leq 7$.}   
			\label{BH}
			\renewcommand{\arraystretch}{1.3} 
			\begin{tabular}{|c|c|}
				\hline \begin{tabular}{c} 
					Alphabet \\
				\end{tabular} & Parameters \\
				\hline $\mathbb{Z}_2$ & $B H(2,2)$ \\
				\hline $\mathbb{Z}_3$ & $B H(3,3), B H(6,3), B H(12,3),B H(21,3)$ \\
				\hline $\mathbb{Z}_4$ & $B H(4,4), B H(6,4), B H(10,4),B H(12,4), B H(14,4)$ \\
				\hline $\mathbb{Z}_5$ & $B H(5,5), B H(10,5)$ \\
				\hline $\mathbb{Z}_6$ & $B H(6,6), B H(7,6), B H(9,6),BH(10,6),BH(13,6),BH(14,6)$ \\
				\hline $\mathbb{Z}_7$ &$B H(7,7), B H(14,7)$ \\
				\hline
			\end{tabular}
		\end{center}
	\end{table}
	

	\section{Generalized Quasi-Florentine Rectangles}
	In this section, we redefine the quasi-Florentine rectangle in a more general way and then provide some constructions.
	
	\begin{definition}\label{GFR}
    A matrix $\mathcal{A}$ over $\mathbb{Z}_N$ is said to be a
\textit{generalized quasi-Florentine rectangle} if it satisfies the following two
conditions:	
		\begin{itemize}
			\item[C1:]  Each row contains $n$ distinct symbols, where each
                  symbol occurs exactly once in each row, $2\leq n\leq N$.
			\item[C2:] For any ordered pair $(a, b)$ of two distinct symbols, and
                       any integer $m$ with  $1\leq m<n$, there is at most one
                        row in which $b$ is $m$ steps right of $a$.
		\end{itemize}
	\end{definition}
    Upon augmenting Condition C2 by imposing that the steps are considered circularly, the improved definition is referred to as a \textit{generalized circular quasi-Florentine rectangle}.  Obviously, every generalized circular quasi-Florentine rectangle is a generalized quasi-Florentine rectangle, but not every generalized quasi-Florentine rectangle qualifies as a generalized circular quasi-Florentine rectangle.

    \begin{example}
    	For $N=8$, the following is an example of a generalized quasi-Florentine rectangle with $F_{Q,2}(8)=8$.
    \[ \mathcal{A}=\left[\begin{array}{llllllll}
			1&2&4&3&6&7 \\
			0&3&5&2&7&6 \\
			3&0&6&1&4&5\\
			5&6&0&7&2&3\\
			2&1&7&0&5&4\\
			7&4&2&5&0&1\\
			6&5&3&4&1&0\\
			4&7&1&6&3&2\\
		\end{array}\right]_{8 \times 6}.\]
    \end{example}
    
\begin{remark}
    The only difference among the Florentine rectangle \cite{avik2021asymptotically}, the quasi-Florentine rectangle \cite{Avik2024}, and the generalized quasi-Florentine rectangle is that, in the Florentine rectangle,  each row must contain all the elements of $\mathbb{Z}_N$. In the quasi-Florentine rectangle defined in \cite{Avik2024}, one element is missing in each row in $\mathbb{Z}_N$.  In contrast, in the generalized quasi-Florentine rectangle, $N-n$ elements are missing in each row, where $2\leq n\leq N$. In particular, when $n = N$, a generalized quasi-Florentine rectangle becomes a Florentine rectangle \cite{avik2021asymptotically}. When $n = N-1$, a generalized quasi-Florentine rectangle becomes a quasi-Florentine rectangle described in \cite{Avik2024}. To know more about Florentine rectangles and quasi-Florentine rectangles, please see \cite{avik2021asymptotically} and \cite{Avik2024} and the references therein. The distinctions among the circular Florentine rectangle \cite{song1992aspects}, the circular quasi-Florentine rectangle \cite{adhikary2025periodic}, and the generalized circular quasi-Florentine rectangle proposed in this work follow analogously from the non-circular case and will not be elaborated here.
\end{remark}
	
	
	Next, we will present two types of known constructions for the circular Florentine rectangles \cite{song1992aspects} and circular quasi-Florentine rectangles\cite{adhikary2025periodic}.
	
	\begin{lemma}[\cite{song1992aspects}]\label{lemma3}
		Let $N$ be a positive integer, and let $p$ be the smallest prime factor of $N$. Define a rectangle $\mathcal{A} = (a_{i,j})_{0 \leq i < p-1,\ 0 \leq j<N}$, $a_{i, j}=(i+1) \times j(\bmod ~N)$, then $\mathcal{A}$ is a  circular Florentine rectangle of size \( (p-1) \times N \) over $\mathbb{Z}_{N}$.
	\end{lemma}
	
	\begin{lemma}[\cite{adhikary2025periodic}]\label{lemma4}
		Let $\mathcal{A}$ be a matrix of order \( p^n \times (p^n - 1) \) defined as follows:
		\[
		\mathcal{A} = \left[\begin{array}{cccc}
			a_{0,0} & a_{0,1} & \cdots & a_{0, p^n-2} \\
			a_{1,0} & a_{1,1} & \cdots & a_{1, p^n-2} \\
			\vdots & \vdots & \ddots & \vdots \\
			a_{p^n-1,0} & a_{p^n-1,1} & \cdots & a_{p^n-1, p^n-2}
		\end{array}\right]_{p^n \times (p^n-1)},
		\]
		where
		\[
		a_{i,j} =
		\begin{cases}
			\psi(\alpha^j), & i = 0, \\
			\psi(\alpha^j + \alpha^{i-1}), & 0 < i < p^n,
		\end{cases}
		\]
		and $\psi$ is a one-to-one mapping.
		Then the matrix $\mathcal{A}$ is a circular quasi-Florentine rectangle of size \( p^n \times (p^n - 1) \) over $\mathbb{Z}_{p^n}$.
	\end{lemma}
	
	Next, we propose two simple constructions of the generalized quasi-Florentine rectangles with the help of Florentine rectangles and quasi-Florentine rectangles. 
	
	\begin{theorem}\label{T1}
		i) Let $\mathcal{A}$ be a Florentine rectangle of size $F(N) \times N$. 
		By removing either the leftmost or the rightmost $c$ columns of $\mathcal{A}$, we obtain a generalized quasi-Florentine rectangle of size 
		$F_{Q,c}(N) \times (N-c)$, where $F_{Q,c}(N) = F(N)$ and $0\leq c<N-1$.\\
		ii) Let $\mathcal{A}$ be a quasi-Florentine rectangle of size $F_{Q,1}(N) \times (N-1)$. 
		By removing either the leftmost or the rightmost $c-1$ columns of $\mathcal{A}$, we obtain a generalized quasi-Florentine rectangle of size 
		$F_{Q,c}(N) \times (N-c)$, where $F_{Q,c}(N) = F_{Q,1}(N)$ and $1\leq c<N-1$.
	\end{theorem}
	
	\begin{proof}
		i) Let $\mathcal{A}=(a_{i,j})_{0 \leq i < F(N),\ 0 \leq j < N}$ be a Florentine rectangle.  
		By definition, each row $a_i$ is a permutation of $\mathbb{Z}_N$.  
		After removing either the leftmost or the rightmost $c$ columns of $\mathcal{A}$, each row $a_i$ contains $N-c$ distinct elements of $\mathbb{Z}_N$, for $0 \leq i < F(N)$.  
		
		Moreover, recall that for two arbitrary distinct elements $a,b \in \mathbb{Z}_N$ and an arbitrary integer $m$ with $1 \leq m < N$, there is at most one row in which $b$ is $m$ steps right of $a$.
		If we remove $c$ columns (either leftmost or rightmost), then for two distinct elements $a,b \in \mathbb{Z}_N$ and an arbitrary integer $m$ with $1 \leq m < N-c$, there still exists at most one row in which $b$ is $m$ steps right of $a$. Therefore, the property is preserved under the removal of $c$ columns, and the proof of (i) is complete.  
		
		ii) The proof of ii) is similar to that of i), hence omitted.
	\end{proof}
	
	Next, we will propose a new construction of the generalized quasi-Florentine rectangles.
	
	\begin{construction} \label{C1}
		Let $N_1$ and $N_2$ be two positive integers. Let \( \mathcal{A} = (a_{i,j})_{0 \leq i < s,\ 0 \leq j < n} \) be an \( s \times n \) matrix over \( \mathbb{Z}_{N_1} \), and let \( \mathcal{B} = (b_{i,j})_{0 \leq i < t,\ 0 \leq j < m} \) be a \( t \times m \) matrix over \( \mathbb{Z}_{N_2} \).
		
		Define the matrix $\mathbf{U}^{(i)}$ for each \( 0 \leq i < \min\{s,t\} \),
		\[
		\mathbf{U}^{(i)} = 
		\begin{bmatrix}
			a_{i,0} + N_1 b_{i,0} & a_{i,1} + N_1 b_{i,0} & \cdots & a_{i,n-1} + N_1 b_{i,0} \\
			a_{i,0} + N_1 b_{i,1} & a_{i,1} + N_1 b_{i,1} & \cdots & a_{i,n-1} + N_1 b_{i,1} \\
			\vdots & \vdots & \ddots & \vdots \\
			a_{i,0} + N_1 b_{i,m-1} & a_{i,1} + N_1 b_{i,m-1} & \cdots & a_{i,n-1} + N_1 b_{i,m-1}
		\end{bmatrix}.
		\]
		
		Now construct an \( \min\{s,t\} \times nm \) matrix \( \mathcal{D} = (d_{i,j})_{0 \leq i < \min\{s,t\},\ 0 \leq j < nm} \) over \( \mathbb{Z}_{N_1 N_2} \), 
		\[
		d_{i,j} = U^{(i)}_{j_1, j_2} = a_{i,j_2} + N_1 b_{i,j_1}, \quad \text{for } j_1 = \left\lfloor \frac{j}{n} \right\rfloor,\ j_2 = j (\bmod~n), 
		\]
		where $\lfloor x\rfloor$ denotes the largest integer less than or equal to the real number $x$.
	\end{construction}
	
	\begin{theorem}\label{T2}
		Let \( \mathcal{A} \) be a generalized circular quasi-Florentine rectangle of size $\widetilde{F}_{Q,N_1-n}(N_1)\times n$ and \( \mathcal{B} \) be a generalized quasi-Florentine rectangle of size $F_{Q,N_2-m}(N_2)\times n$. Then, the matrix \( \mathcal{D} \) in Construction~\ref{C1} is a generalized quasi-Florentine rectangle of size $F_{Q,N_1N_2-nm}(N_1N_2)\times mn$, where $F_{Q,N_1N_2-nm}(N_1 N_2)=\min\{\widetilde{F}_{Q,N_1-n}(N_1),F_{Q,N_2-m}(N_2)\}$.
	\end{theorem}
    
	\begin{proof}
		We first prove that any two elements in the row \( d_i \) of the matrix \( \mathcal{D} \) are distinct for \( 0 \leq i < s \), i.e., the matrix \( \mathcal{D} \) satisfies condition C1 of Definition \ref{GFR}.
		
		Let \( d_{i,j} \) and \( d_{i,k} \) be two elements in a row \( d_i \) with \( j \neq k \). Then
		\[
		d_{i,j} = U^{(i)}_{j_1, j_2} = a_{i,j_2} + N_1 b_{i,j_1}, \quad
		d_{i,k} = U^{(i)}_{k_1, k_2} = a_{i,k_2} + N_1 b_{i,k_1},
		\]
		where $j_1 = \left\lfloor \frac{j}{n} \right\rfloor$, $j_2 = j (\bmod~ n)$, $k_1 = \left\lfloor \frac{k}{n} \right\rfloor$, and $k_2 = k (\bmod~ n)$.
		
        Since $j\neq k$, the discussion for condition (1) is divided into three cases:
        
		Case 1: \( j_1 = k_1 \) and \( j_2 \neq k_2 \). Since \( b_{i,j_1} = b_{i,k_1} \), \( a_{i,j_2} \neq a_{i,k_2} \), and \( 0 \leq a_{i,j_2}, a_{i,k_2} < N_1 \), we have \( d_{i,j} \neq d_{i,k} \).
		
		Case 2: \( j_1 \neq k_1 \) and \( j_2 = k_2 \). This is similar to Case 1 and hence omitted.
		
		Case 3: \( j_1 \neq k_1 \) and \( j_2 \neq k_2 \). Since \( a_{i,j_2} \neq a_{i,k_2} \) and \( b_{i,j_1} \neq b_{i,k_1} \), and given that \( 0 \leq a_{i,j_2}, a_{i,k_2} < N_1 \), \( 0 \leq b_{i,j_1}, b_{i,k_1} < N_2 \), we have
		\[
		d_{i,j} - d_{i,k} = (a_{i,j_2} - a_{i,k_2}) + N_1(b_{i,j_1} - b_{i,k_1}) \neq 0.
		\]
		Thus, \( d_{i,j} \neq d_{i,k} \) in all cases.
		
		Next, we will prove that the matrix \( \mathcal{D} \) satisfies the condition C2 of Definition \ref{GFR}.
		   Assume that there are two elements $a$ and $b$, where $b$ is $\tau$ steps right of $a$ in two rows of $\mathcal{A}$, say, $u$ and $v$, i.e.,
		\begin{align}
			&d_{u,j} = a,\quad d_{u,j+\tau} = b, \label{21} \\
			&d_{v,k} = a,\quad d_{v,k+\tau} = b, \label{22}
		\end{align}
		where \( a, b \in \mathbb{Z}_{N_1N_2} \), $0<\tau<mn$, and \( 0 \leq j + \tau, k + \tau < nm \).
		
		Let \(\tau=n\tau_1+\tau_2\), where \( \tau_1 = \left\lfloor \frac{\tau}{n} \right\rfloor,\ \tau_2 = \tau (\bmod~ n) \). Define
		\[
		j_1 = \left\lfloor \frac{j}{n} \right\rfloor,\ j_2 = j (\bmod~ n), \quad
		k_1 = \left\lfloor \frac{k}{n} \right\rfloor,\ k_2 = k (\bmod~ n).
		\]
		
		To proceed, we consider two distinct cases based on the value of $\tau_2 + j_2$ ($0 \leq \tau_2 + j_2 < n$ and $n < \tau_2 + j_2 < 2n$):
		
		Case 1: \( 0 \leq \tau_2 + j_2 < n \). We have
		\[
		d_{u,j} = a_{u,j_2} + N_1 b_{u,j_1} = a, \quad
		d_{u,j+\tau} = a_{u,j_2+\tau_2} + N_1 b_{u,\left\lfloor \frac{j+\tau}{n} \right\rfloor} = b,
		\]
		\[
		d_{v,k} = a_{v,k_2} + N_1 b_{v,k_1} = a, \quad
		d_{v,k+\tau} = a_{v,k_2+\tau_2} + N_1 b_{v,\left\lfloor \frac{k+\tau}{n} \right\rfloor} = b.
		\]
		
		Since the ranges of \( a_{i,j} \) and \( b_{i,j} \) ensure uniqueness in representation, equalities \( d_{u,j} = d_{v,k} \) and \( d_{u,j+\tau} = d_{v,k+\tau} \) imply:
		\[
		a_{u,j_2} = a_{v,k_2},\ b_{u,j_1} = b_{v,k_1},\quad
		a_{u,j_2+\tau_2} = a_{v,k_2+\tau_2},\ b_{u,\left\lfloor \frac{j+\tau}{n} \right\rfloor} = b_{v,\left\lfloor \frac{k+\tau}{n} \right\rfloor}.
		\]
		
		Since \( \tau \neq 0 \), we consider the following three subcases:
		
		- If \( \tau_1 = 0\) and \( \tau_2 \neq 0 \), then this contradicts the generalized circular quasi-Florentine rectangles property of \( \mathcal{A} \).
		
		- If \( \tau_1 \neq 0\) and \( \tau_2 = 0 \), then $0<\left\lfloor \frac{j+\tau}{n} \right\rfloor,\left\lfloor \frac{k+\tau}{n} \right\rfloor<m$ by $0<j+\tau,k+\tau<nm$. This contradicts the generalized quasi-Florentine rectangles property of \( \mathcal{B} \).
		
		- If \( \tau_1 \neq 0\) and \( \tau_2 \neq 0 \), then this contradicts both \( \mathcal{A} \) being generalized circular quasi-Florentine rectangles and \( \mathcal{B} \) being generalized quasi-Florentine rectangles.
		
		Therefore, the assumption that both (\ref{21}) and (\ref{22}) hold leads to a contradiction.
		
		Case 2: \( n < \tau_2 + j_2 < 2n \). We have
		\[
		d_{u,j+\tau} = a_{u,j_2+\tau_2 - n} + N_1 b_{u,\left\lfloor \frac{j+\tau}{n} \right\rfloor},\quad
		d_{v,k+\tau} = a_{v,k_2+\tau_2 -n} + N_1 b_{v,\left\lfloor \frac{k+\tau}{n} \right\rfloor}.
		\]
		
		The proof proceeds analogously to Case 1 and is thus omitted.
		
		This completes the proof.
	\end{proof}
	
	Using the generalized quasi-Florentine rectangles, circular Florentine rectangles, and circular quasi-Florentine rectangles given in Theorem~\ref{T1}, Lemmas~\ref{lemma3} and \ref{lemma4}, we obtain the following result, which can be proved in a manner similar to the proof of Theorem~\ref{T2}. Therefore, the proof is omitted.
	
	\begin{corollary} \label{C3}
         Let \( \mathcal{A} \) be a circular Florentine rectangle of order $(p_1-1)\times N_1$ over $\mathbb{Z}_{N_1}$ in Lemma \ref{lemma3}, where $N_1$ is a positive integer and $p_1$ is the smallest prime factor of $N_1$.
		\begin{enumerate}
			\item[i)] Let \( \mathcal{B} \) be a generalized quasi-Florentine rectangle of size $p^{n}\times (p^{n}-c)$ over $\mathbb{Z}_{p^n}$. By Theorem \ref{T2}, $\mathcal{D}$ is a generalized quasi-Florentine rectangle of size $\min\{p_1-1, p^{n}\}\times N_1(p^{n}-c)$ over $\mathbb{Z}_{N_1 p^n}$, where $n$ is a positive integer, $p$ is a prime, and $1\le c<p^n-1$.
			
			\item[ii)]  Let \( \mathcal{B} \) be a generalized quasi-Florentine rectangle of size $p^{n}\times (p^{n}+1-c)$ over $\mathbb{Z}_{p^n+1}$. By Theorem \ref{T2}, $\mathcal{D}$ is a generalized quasi-Florentine rectangle of size $\min\{p_1-1, p^{n}\}\times N_1(p^{n}+1-c)$ over $\mathbb{Z}_{N_1(p^{n}+1)}$,  where $n$ is a positive integer, $p$ is a prime, and $1\le c<p^n$.
		\end{enumerate}
	\end{corollary}
	
	\begin{corollary} \label{C2}
    Let \( \mathcal{A} \) be a circular quasi-Florentine rectangle of order $p^{n}\times (p^{n}-1)$ over $\mathbb{Z}_{p^n}$ in Lemma \ref{lemma4}, where $p$ is a prime number and $n$ is a positive integer.
    \begin{enumerate}
			\item[i)]  Let \( \mathcal{B} \) be a generalized quasi-Florentine rectangle of size $(p_1-1)\times (N_1-c)$ over $\mathbb{Z}_{N_1}$. By Theorem \ref{T2}, $\mathcal{D}$ is a generalized quasi-Florentine rectangle of size $\min\{p_1-1, p^{n}\} \times (N_1-c)(p^{n}-1)$ over $\mathbb{Z}_{N_1p^{n}}$,
            where $N_1$ is a positive integer, $p_1$ is the smallest prime factor of $N_1$, and $0\le c<N_1-1$.
			
			\item[ii)]  Let \( \mathcal{B} \) be a generalized quasi-Florentine rectangle of size $p_1^{n_1}\times (p_1^{n_1}-c)$ over $\mathbb{Z}_{p_1^{n_1}}$.
             By Theorem \ref{T2}, $\mathcal{D}$ is a generalized quasi-Florentine rectangle of size $\min\{p^{n}, p_1^{n_1}\} \times (p^{n}-1)(p_1^{n_1}-c)$ over $\mathbb{Z}_{p^np_1^{n_1}}$, where $p_1$ is a prime, $n_1$ is a positive integer, and $1\le c<p_1^{n_1}-1$.
			
           \item[iii)] Let \( \mathcal{B} \) be a generalized quasi-Florentine rectangle of size $p_1^{n_1}\times (p_1^{n_1}+1-c)$ over $\mathbb{Z}_{p_1^{n_1}+1}$. By Theorem \ref{T2}, $\mathcal{D}$ is a generalized quasi-Florentine rectangle of size $\min\{p^n, p_1^{n_1}\} \times (p^n-1)(p_1^{n_1}+1-c)$ over $\mathbb{Z}_{p^n(p_1^{n_1}+1)}$, where  $p_1$ is a prime, $n_1$ is a positive integer, and $1\le c<p_1^{n_1}$.
		\end{enumerate}

	\end{corollary}
	

	\begin{example}
		Let the matrix $\mathcal{A}$ be a circular Florentine rectangle of order $6\times 7$ over $\mathbb{Z}_7$, and let the matrix $\mathcal{B}$ be a quasi-Florentine rectangle of order $9\times 8$ over $\mathbb{Z}_9$. 
		\[\mathcal{A}=\left[\begin{array}{lllllll}
			0 & 1 & 2 & 3 & 4 & 5 & 6  \\
			0 & 2 & 4 & 6 & 1 & 3 &5 \\
			0 & 3 & 6 & 2 & 5 & 1 & 4  \\
			0 & 4 & 1 & 5 & 2 & 6 & 3  \\
			0 & 5 & 3 & 1 & 6 & 4 & 2 \\
			0 & 6 & 5 & 4 & 3 & 2 & 1 
		\end{array}\right]_{6 \times 7},~~~ \mathcal{B}=\left[\begin{array}{llllllll}
			1&2&4&3&6&7&5&8 \\
			0&3&5&2&7&6&4&8 \\
			3&0&6&1&4&5&7&8\\
			5&6&0&7&2&3&1&8\\
			2&1&7&0&5&4&6&8\\
			7&4&2&5&0&1&3&8\\
			6&5&3&4&1&0&2&8\\
			4&7&1&6&3&2&0&8\\
		\end{array}\right]_{8 \times 8}.\]
		It is easy to check that the matrix $\mathcal{D}=[d_0;d_1;d_2;d_3;d_4;d_5]$ in Construction \ref {C1} is a generalized quasi-Florentine rectangle over $\mathbb{Z}_{63}$, with $F_{Q,7}(63)=6$, where
		\[\begin{aligned}
			d_0=(&7~8~9~10~11~12~13~14~15~16~17~18~19~20~28~29~30~31~32~33~34~21~22~23~24~25~26~27~42\\&43~44~45~46~47~48~49~50~51~52~53~54~55~35~36~37~38~39~40~41~56~57~58~59~60~61~62);
		\end{aligned}\]
		\[\begin{aligned}
			d_1=(&0~2~4~6~1~3~5~21~23~25~27~22~24~26~35~37~39~41~36~38~40~14~16~18~20~15~17~19~49~51\\&53~55~50~52~54~42~44~46~48~43~45~47~28~30~32~34~29~31~33~56~58~60~62~57~59~61);
		\end{aligned}\]
		\[\begin{aligned}
			d_2=(&21~24~27~23~26~22~25~0~3~6~2~5~1~4~42~45~48~44~47~43~46~7~10~13~9~12~8~11~28~31~34\\&30~33~29~32~35~38~41~37~40~36~39~49~52~55~51~54~50~53~56~59~62~58~61~57~60);
		\end{aligned}\]
		\[\begin{aligned}
			d_3=(&35~39~36~40~37~41~38~42~46~43~47~44~48~45~0~4~1~5~2~6~3~49~53~50~54~51~55~52~14~18\\&15~19~16~20~17~21~25~22~26~23~27~24~7~11~8~12~9~13~10~56~60~57~61~58~62~59);
		\end{aligned}\]
		\[\begin{aligned}
			d_4=(&14~19~17~15~20~18~16~7~12~10~8~13~11~9~49~54~52~50~55~53~51~0~5~3~1~6~4~2~35~40~38\\&36~41~39~37~28~33~31~29~34~32~30~42~47~45~43~48~46~44~56~61~59~57~62~60~58);
		\end{aligned}\]
		\[\begin{aligned}
			d_5=(&49~55~54~53~52~51~50~28~34~33~32~31~30~29~14~20~19~18~17~16~15~35~41~40~39~38~37~36\\&~0~6~5~4~3~2~1~7~13~12~11~10~9~8~21~27~26~25~24~23~22~56~62~61~60~59~58~57).
		\end{aligned}\]
	\end{example}
	
	\begin{remark}
		Among the existing Florentine rectangle parameters with a mathematical construction method, the number of rows in the matrix is relatively small, typically given by $F(N)=\max\{4,p-1\}$, where $p$ is the smallest prime factor of $N$. For circular Florentine rectangle parameters with a mathematical construction method, $\widetilde{F}(N)=p-1$ \cite{song1992aspects}. In particular, when $N$ is even, $\widetilde{F}(N)=1$.  For quasi-Florentine rectangles, when $N=p^n$ and $N=p^n+1$, it has been shown that $F_{Q,1}(N)=p^n$.  For other values of $N$, the number of rows remains $F_{Q,1}(N)=F(N)=\max\{4,p-1\}$ \cite{Avik2024}. Thus, for certain other values of $N$ that are not in the form $p$, $p-1$, $p^n$, or $p^n+1$, the generalized quasi-Florentine rectangles proposed in this work can achieve a greater number of rows. Table~\ref{table_2} presents specific examples of such parameters.
	\end{remark}

	\begin{table}[htp]
		\caption{Parameters of the Proposed Generalized Quasi-Florentine Rectangles}
		\label{table_2}
		\renewcommand{\arraystretch}{1.3} 
		\begin{center}
			\begin{tabular}{|c|c|c|c|c|c|}
				\hline
				$N$ & $L$ & $\widetilde{F}(N)$& $F(N)$ & $F_{Q,1}(N)$ & $F_{Q,N-L}(N)$ \\
				\hline
				$2^4\times(3^2+1)$	 & $(2^4-1)\times3^2$  & $1$ & $4$  & $4$ & $9$ \\
				\hline
				$2^5\times(5^2+1)$	& $(2^5-1)\times 5^2$& $1$ & $4$ & $4$  & $25$ \\
				\hline
				$3^3\times (7^2+1)$  	& $(3^3-1)\times 7^2$& $1$ & $4$ & $4$ &  $27$\\
				\hline
				$37\times (2^5+1)$ & $37\times 2^5$ &  $2$& $4$ &  $4$ &  $32$\\
				\hline
				$47\times 7^2$ & $47\times (7^2-1)$ & $6$ & $6$& $6$ & $46$ \\
				\hline
				$61\times 2^6$	& $61\times (2^6-1)$ &  $1$& $4$ & $4$ & $60$ \\
				\hline
			\end{tabular}
		\end{center}
	\end{table}

	\begin{lemma} \label{LemmaGFR}
		Let \( \mathcal{A} \) be a generalized quasi-Florentine rectangle of order $F_{Q,N-n}(N)\times n$ over $\mathbb{Z}_N$, and let \( a_i \) be the \( i \)-th row of \( \mathcal{A} \). For any pair of distinct indices \( 0 \leq i \ne p < F_{Q,N-n}(N) \), the equation \( a_{i,j} = a_{p,j+\tau} \) has at most one solution over the range  \(0\leq j<n-\tau\), for every  \( 0 \leq \tau < n \).
	\end{lemma}
	\begin{proof}
		Suppose \( 0 \leq i \ne p < F_{Q,N-n}(N) \), and that there exists some \( 0 < \tau < n \) such that the equation \( a_{i,j} = a_{p,j+\tau} \) (with \( 0 \leq j + \tau < n \)) has two solutions, say \( j_1 \) and \( j_2 \). Then we have
		\[
		a_{i,j_1} = a_{p,j_1+\tau}, \quad a_{i,j_2} = a_{p,j_2+\tau}.
		\]
		This implies
		\[
		(a_{i,j_1}, a_{i,j_2}) = (a_{p,j_1+\tau}, a_{p,j_2+\tau}),
		\]
		which contradicts the definition of a generalized quasi-Florentine rectangle. Therefore, for any \( 0 \leq \tau < n \), \( 0 \leq j + \tau < n \), and \( 0 \leq i \ne p < F_{Q,N-n}(N) \), the equation \( a_{i,j} = a_{p,j+\tau} \) has at most one solution.
	\end{proof}
	
	%
	%

	\section{A Construction of aperiodic DRCS sets} \label{5}
	In this section, we propose some new aperiodic DRCS sets that are asymptotically optimal with respect to the bound in Lemma \ref{lem-wangbound}.  First, we present the construction framework.
	
	\begin{construction}\label{c1}
		Let $N\geq 2$ be a positive integer such that a $K\times L$ generalized quasi Florentine rectangle $\mathcal{A}$ exists over $\mathbb{Z}_N$, where $K=F_{Q,N-L}(N)$. Let $a_{i,j}$ denotes the $j$-th element of the $i$-th row of $\mathcal{A}$. For a positive integer $r$, let $\mathbf{B}$ be a Butson-type Hadamard matrix of order $N$ over $\mathbb{Z}_r$, given by
		\begin{align*}
			\mathbf{B}=\left[\begin{array}{c}
				\mathbf{b}_0 \\
				\mathbf{b}_1 \\
				\vdots \\
				\mathbf{b}_{N-1}
			\end{array}\right]=\left[\begin{array}{cccc}
				\omega_r^{b_{0,0}} & \omega_r^{b_{0,1}} & \cdots & \omega_r^{b_{0, N-1}} \\
				\omega_r^{b_{1,0}} & \omega_r^{b_{1,1}} & \cdots & \omega_r^{b_{1, N-1}} \\
				\vdots & \vdots & \ddots & \vdots \\
				\omega_r^{b_{N-1,0}} & \omega_r^{b_{N-1,1}} & \cdots & \omega_r^{b_{N-1, N-1}}
			\end{array}\right].
		\end{align*}
		Define a DRCS set $\mathcal{C}=\{\mathbf{C}^{(0)},\mathbf{C}^{(1)},\cdots,\mathbf{C}^{(K-1)}\}$, where
		
		\begin{equation}
			\mathbf{C}^{(k)}=\left[\begin{array}{c}
				\mathbf{c}_0^{(k)} \\
				\mathbf{c}_1^{(k)}\\
				\vdots \\
				\mathbf{c}_{N-1}^{(k)}
			\end{array}\right]=\left[\begin{array}{c}
				c_{0,0}^{(k)}, ~~~c_{0,1}^{(k)}, ~\cdots, ~~~c_{0, L-1}^{(k)} \\
				c_{1,0}^{(k)}, ~~~c_{1,1}^{(k)}, ~\cdots, ~~~c_{1, L-1}^{(k)} \\
				~~~		\vdots \\
				c_{N-1,0}^{(k)}, c_{N-1,1}^{(k)}, \cdots, c_{N-1, L-1}^{(k)}
			\end{array}\right],
		\end{equation}
		and
		\begin{align*}
			c_{m, n}^{(k)}=\omega_r^{b_{a_{k,n},m}}
		\end{align*}
		for $0\leq k<K$, $0 \leq m<N$, and $0 \leq n<L$.
	\end{construction}
	
	For the sequence sets generated by Construction \ref{c1}, we have the following result.

	\begin{theorem}\label{thc2} 
		The sequence set $\mathcal{C}$ given by Construction \ref{c1} is an aperiodic DRCS set with parameters  $\left(K, N, L, N, \Pi\right)$-DRCS, where $\Pi=(-L,L)\times (-L, L)$. 
	\end{theorem}
	\begin{proof}
		According to the definition of DRCS set, we divide the proof into two cases: auto-AF and cross-AF.
		
		Case 1 (auto-$\hat{AF}$): For any $0 \leq k<K$, we have
		\begin{align}\label{coneq2}
			\begin{split}
				\hat{AF}_{\mathbf{C}^{(k)}}(\tau, \nu) & =\sum_{m=0}^{N-1} \sum_{n=0}^{L-1-\tau} 
				c_{m, n}^{(k)} (c_{m, n+\tau}^{(k)})^{*}\omega_{L}^{\nu n}\\ 
				&=\sum_{m=0}^{N-1}\sum_{n=0}^{L-1-\tau}
				\omega_r^{b_{a_{k,n},m}}\omega_r^{-b_{a_{k,n+\tau},m}}\omega_{L}^{\nu n}\\
				&=\sum_{n=0}^{L-1-\tau}\omega_{L}^{\nu n}\sum_{m=0}^{N-1}\omega_r^{b_{a_{k,n},m}-b_{a_{k,n+\tau},m}}.
			\end{split}
		\end{align}
		
		When $\tau=0$ and $\nu=0$, we have
		$$
		\hat{AF}_{\mathbf{C}^{(k)}}(0, 0)=NL.
		$$
		
		When $\tau=0$ and $0<|\nu|<N-1$, we find that
		$$
		\left|\hat{AF}_{\mathbf{C}^{(k)}}(0, \nu)\right|=\left|N\sum_{n=0}^{L-1}\omega_{L}^{\nu n}\right|= 0.
		$$
		
		When $0<|\tau|<L$ and $0\leq |\nu|<L$, since $a_{k,n}\neq a_{k,n+\tau}$, the vectors $\mathbf{b}_{a_{k,n}}$ and $\mathbf{b}_{a_{k,n+\tau}}$ are two distinct rows of a Butson-type Hadamard matrix, and they are orthogonal to each other. Therefore, we have
		\begin{align*}
			\sum_{m=0}^{N-1}\omega_r^{b_{a_{k,n},m}-b_{a_{k,n+\tau},m}} =0.
		\end{align*}
		Thus
		\begin{align*}
			\left|\hat{AF}_{\mathbf{C}^{(k)}}(\tau, \nu)\right|=0.
		\end{align*}
		
		Case 2 (cross-$\hat{AF}$): For any $0 \leq k_1\neq k_2<K$, we have
		\begin{align*}
			\begin{split}
				\hat{A F}_{\mathbf{C}^{(k_1)},\mathbf{C}^{(k_2)}}(\tau, \nu) & 
				=\sum_{m=0}^{N-1} \sum_{n=0}^{L-1-\tau} 
				c_{m, n}^{(k_1)} (c_{m, n+\tau}^{(k_2)})^{*}\omega_{L}^{\nu n}\\ 
				&=\sum_{m=0}^{N-1}\sum_{n=0}^{L-1-\tau}
				\omega_r^{b_{a_{k_1,n},m}}\omega_r^{-b_{a_{k_2,n+\tau},m}}\omega_{L}^{\nu n}\\
				&=\sum_{n=0}^{L-1-\tau}\omega_{L}^{\nu n}\sum_{m=0}^{N-1}\omega_r^{b_{a_{k_1,n},m}-b_{a_{k_2,n+\tau},m}}.
			\end{split}
		\end{align*}
		Since
		\begin{align}\label{solution_2}
			a_{k_1,n}=a_{k_2,n+\tau}
		\end{align}
		has at most one solution over $0\leq n<L-\tau$ for $0\leq \tau<L$ and $0 \leq k_1 \neq k_2<K$ by the Lemma \ref{LemmaGFR}. If (\ref{solution_2}) does not have a solution, then
		\begin{align*}
			\hat{A F}_{\mathbf{C}^{(k_1)},\mathbf{C}^{(k_2)}}(\tau, \nu)=0.
		\end{align*}
		If (\ref{solution_2}) has exactly one solution, we denote it as $n'$. The vectors $\mathbf{b}_{a_{k_1,n}}$ and $\mathbf{b}_{a_{k_2,n+\tau}}$ are two distinct rows of a Butson-type Hadamard matrix, and they are orthogonal to each other. Therefore, we have
		\begin{align*}
			\hat{A F}_{\mathbf{C}^{(k_1)},\mathbf{C}^{(k_2)}}(\tau, \nu)
			&=N\omega_{L}^{\nu n'}+\sum_{n=0,n\neq n'}^{L-1-\tau}\omega_{L}^{\nu n}
			\sum_{m=0}^{N-1}\omega_r^{b_{a_{k_1,n},m}-b_{a_{k_2,n+\tau},m}} \nonumber \\
			&=N\omega_{L}^{\nu n'}.
		\end{align*}
		Namely, $|\hat{AF}_{\mathbf{C}^{(k_1)},\mathbf{C}^{(k_2)}}(\tau, \nu)|=N$ for all $k_1 \neq k_2$, $0\leq |\tau|<Z_x$ and $0 \leq \nu<L$. This completes the proof.
	\end{proof}
	
	\begin{remark}\label{remark5}
		Notably, when the generalized quasi-Florentine rectangle reduces to the quasi-Florentine rectangle as defined in \cite{Avik2024}, the proposed \( (K, N, N-1, N, \Pi) \)-DRCS reduces to the construction by Wang \textit{et al.} in \cite{wang2025doppler}, where \( \Pi = (-N+1, N-1) \times (-N+1, N-1) \). Furthermore, when the Butson-type Hadamard matrix is selected as the DFT matrix and the generalized quasi-Florentine rectangle chooses a circular Florentine rectangle, the proposed \( (K, N, N, N, \Pi) \)-DRCS reduces to the construction by Shen \textit{et al.} in \cite{shen2024}, where \( \Pi = (-N, N) \times (-N, N) \).
	\end{remark}
	
	Next, we analyze the DRCS set using some specific types of generalized quasi-Florentine rectangles provided in Corollaries \ref{C3} and \ref{C2} under Theorem \ref{thc2}. Similar analyses can be carried out for other parameter choices, and the corresponding results can be established through arguments analogous to those used in the proof of Theorem \ref{thc2}. Therefore, we omit the detailed proofs.

	\begin{corollary}\label{CT2}
		The set $\mathcal{C}$ constructed in Theorem \ref{thc2}, based on the generalized quasi-Florentine rectangles provided in Corollaries \ref{C3} and \ref{C2}, forms an aperiodic DRCS set with the following parameters:
		\begin{enumerate}
			\item[i)]  When $\mathcal{A}$ is a generalized quasi-Florentine rectangle of size $K\times N_1(p^{n}-c)$ over $\mathbb{Z}_{N_1p^n}$ in Corollary \ref{C3}, $\mathcal{C}$ is a $(K,N_1 p^{n}, N_1(p^{n}-c), N_1 p^{n}, \Pi)$-DRCS, where $\Pi=(-N_1(p^{n}-c),N_1(p^{n}-c))\times (-N_1(p^{n}-c), N_1(p^{n}-c))$, $K=\min\{p_1-1, p^{n}\}$, $p$ is a prime, $p_1$ is the smallest prime factor of $N_1$, $N_1$ and $n$ are positive integers, and $1\le c <p^n-1$.
			\item[ii)] When $\mathcal{A}$ is a generalized quasi-Florentine rectangle of size $K\times N_1(p^{n}+1-c)$ over $\mathbb{Z}_{N_1(p^{n}+1)}$ in Corollary \ref{C3}, $\mathcal{C}$ is a $(K,N_1(p^{n}+1), N_1(p^{n}+1-c), N_1(p^{n}+1), \Pi)$-DRCS, where $\Pi=(-N_1(p^{n}+1-c),N_1(p^{n}+1-c))\times (-N_1(p^{n}+1-c), N_1(p^{n}+1-c))$, $K=\min\{p_1-1, p^{n}\}$, $p$ is a prime, $p_1$ is the smallest prime factor of $N_1$, $N_1$ and $n$ are positive integers, and $1\le c<p^n$.  
			\item[iii)] When $\mathcal{A}$ is a generalized quasi-Florentine rectangle of size $K\times (N_1-c)(p^{n}-1)$ over $\mathbb{Z}_{N_1p^n}$ in Corollary \ref{C2}, $\mathcal{C}$ is a $(K,N_1p^{n}, (N_1-c)(p^{n}-1), N_1p^{n}, \Pi)$-DRCS, where $\Pi=(-(N_1-c)(p^{n}-1),(N_1-c)(p^{n}-1))\times (-(N_1-c)(p^{n}-1), (N_1-c)(p^{n}-1))$, $K=\min\{p_1-1, p^{n}\}$, $p$ is a prime, $p_1$ is the smallest prime factor of $N_1$, $N_1$ and $n$ are positive integers, and $0\le c<N_1-1$.
			\item[IV)] When $\mathcal{A}$ is a generalized quasi-Florentine rectangle of size $K\times (p^n-1)(p_1^{n_1}-c)$ over $\mathbb{Z}_{p^np_1^{n_1}}$ in Corollary \ref{C2}, $\mathcal{C}$ is a $(K,p^{n}p_1^{n_1}, (p^{n}-1)(p_1^{n_1}-c), p^{n}p_1^{n_1}, \Pi)$-DRCS, where $\Pi=(-(p^{n}-1)(p_1^{n_1}-c),(p^{n}-1)(p_1^{n_1}-c))\times (-(p^{n}-1)(p_1^{n_1}-c), (p^{n}-1)(p_1^{n_1}-c))$, $K=\min\{p^{n}, p_1^{n_1}\}$, $p$ and $p_1$ are primes, $n$ and $n_1$ are positive integers, and $1\le c<p^{n_1}-1$.
			\item[V)] When $\mathcal{A}$ is a generalized quasi-Florentine rectangle of size $K\times (p^n-1)(p_1^{n_1}+1-c)$ over $\mathbb{Z}_{p^n(p_1^{n_1}+1)}$ in Corollary \ref{C2}, $\mathcal{C}$ is a $(K,p^n(p_1^{n_1}+1), (p^n-1)(p_1^{n_1}+1-c), p^n(p_1^{n_1}+1), \Pi)$-DRCS, where $\Pi=(-(p^n-1)(p_1^{n_1}+1-c),(p^n-1)(p_1^{n_1}+1-c)\times (-(p^n-1)(p_1^{n_1}+1-c), (p^n-1)(p_1^{n_1}+1-c))$, $K=\min\{p^{n}, p_1^{n_1}\}$, $p$ and $p_1$ are primes, $n$ and $n_1$ are positive integers, and $1\le c<p_1^{n_1}$.
		\end{enumerate}
	\end{corollary}
	
	In the following, we show an example to illustrate the proposed construction.
	\begin{example}\label{ex1}
		Let $N=63$, $L=56$, $K=6$, $\mathcal{A}$ be a generalized quasi-Florentine rectangle of over $\mathbb{Z}_{63}$ with $F_{Q,7}(63)=6$,  and let $\mathbf{B}$ be a $BH(63,3)$ by Lemma \ref{BH_lemma}.
		The sequence set $\mathcal{C}$ is constructed by using Construction 1. It is easy to check that $\mathcal{C}$ forms an aperiodic $(6,63,56,63,\Pi)$-DRCS set over $\mathbb{Z}_3$, where $\Pi = (-56,56)\times (-56,56)$. The optimality factor in this case is $\hat{\rho} = 1.5$. A glimpse of the aperiodic auto-AF and cross-AF of the DRCSs $\mathcal{C}$ can be seen in Fig. \ref{fig:ex1}. 
		\begin{figure}[htp]
			\centering
			\includegraphics[width=1\linewidth]{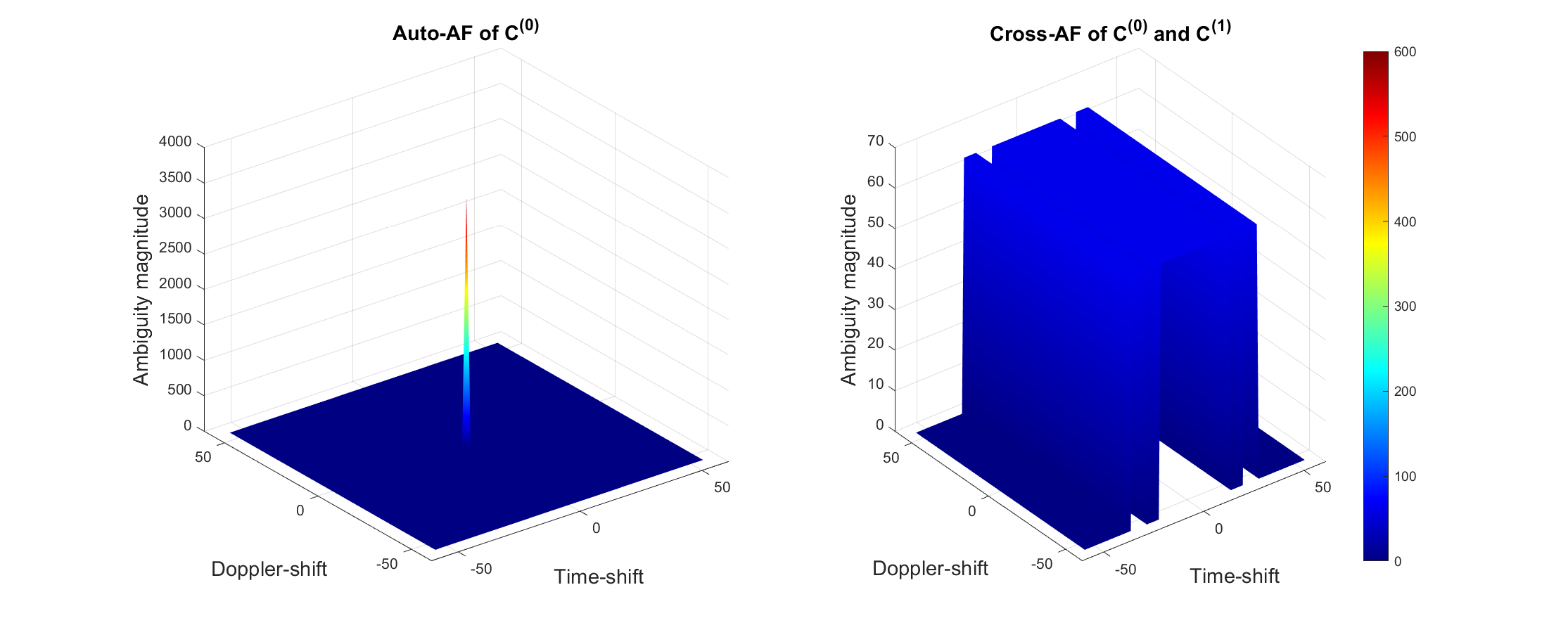}
			\caption{A glimpse of the aperiodic auto-AF and cross-AF of the sequence set $\mathcal{C}$ in Example \ref{ex1}.}
			\label{fig:ex1}
		\end{figure}
	\end{example}
	
	\subsection{ Discussion on Optimality}
	In this subsection, we examine the optimality of the proposed aperiodic DRCS set $\mathcal{C}$ mentioned
	in  Theorem \ref{thc2} and Corollary \ref{CT2}. 
	
	\begin{theorem} \label{T4}
		Let $\mathcal{C}$ be an aperiodic DRCS set with parameters  $\left(K, N, L, N, \Pi\right)$-DRCS constructed according to Theorem~\ref{thc2}, where $\Pi=(-L,L)\times (-L, L)$.
		Then, $\mathcal{C}$ is asymptotically optimal with respect to the bound in Lemma \ref{lem-wangbound} if $K > \frac{3N}{L} $, $\lim\limits_{L \rightarrow \infty}\frac{L}{N}=1$, and $\lim\limits_{L \rightarrow \infty}\frac{1}{K}=0$.
	\end{theorem}
	
	\begin{proof}
		Based on Theorem \ref{thc2}, the set \( \mathcal{C} \) is an aperiodic \((K,N,L,N,\Pi)\)-DRCS sequence set with \(\Pi = (-L, L) \times (-L, L)\). The optimality factor \( \hat{\rho} \) is given by
		\[
		\hat{\rho} = \frac{N}{\sqrt{NL \left( 1 - 2 \sqrt{\frac{N}{3KL}} \right)}}
		\]
		since $K > \frac{3N}{L} $.
		
		When $\lim\limits_{L \rightarrow \infty}\frac{L}{N}=1$ and $\lim\limits_{L \rightarrow \infty}\frac{1}{K}=0$, we have 
		\[\lim\limits_{L \rightarrow \infty}\hat{\rho}=\lim\limits_{L \rightarrow \infty}\frac{N}{\sqrt{NL \left( 1 - 2 \sqrt{\frac{N}{3KL}} \right)}}=\lim\limits_{L \rightarrow \infty}\frac{1}{\sqrt{\frac{L}{N} \left( 1 - 2 \sqrt{\frac{N}{3KL}} \right)}} =1.\] 
		Thus, the aperiodic DRCS set \( \mathcal{C} \) is asymptotically optimal.	
	\end{proof}
	
	Next, we analyze the asymptotically optimal aperiodic DRCS set in Corollary \ref{CT2}.
	\begin{corollary}
		Let $\mathcal{C}$ be the aperiodic DRCS set in Corollary \ref{CT2}. Then
		\begin{enumerate}
			\item[i)] the DRCS set $\mathcal{C}$ of Corollary \ref{CT2} i), ii), and iii) are asymptotically optimal based on the bound given in Lemma \ref{lem-wangbound}, when $K\geq 4+c$, $q^n=N_1\pm k$, and  $\lim\limits_{N_1 \rightarrow \infty}\frac{1}{p_1}=0$, where $k$ is a constant and $p_1$ is the smallest prime factor of $N_1$.
			\item[ii)] the DRCS set $\mathcal{C}$ of Corollary \ref{CT2} IV) and V) are asymptotically optimal based on the bound given in Lemma \ref{lem-wangbound}, when $K\geq 4+c$, $p_1=p\pm k$, where $k$ is a constant.
		\end{enumerate}
	\end{corollary}
	
	\begin{proof}
		i) By Corollary \ref{CT2} i), we have $K=\min\{p_1-1, p^{n}\}$, $L=N_1(p^{n}-c)$, and $N=N_1p^{n}$.  When \( K \geq 4+c \), it holds that
		\[
		K =\min\{p_1-1, p^{n}\}>3+\frac{3c}{p^{n}-c}=\frac{3p^{n}}{p^{n}-c}.
		\]
		
		Consequently, we obtain
		\[
		K =\min\{p_1-1, p^{n}\}> \frac{3 N_1p^{n}}{N_1(p^{n}-c)} = \frac{3N}{L}.
		\]
		
		It is obvious that
		\[
		\lim_{L \to \infty} \frac{L}{N} = \lim_{N_1(p^{n}-c) \to \infty} \frac{N_1(p^{n}-c)}{N_1p^{n}} = 1.
		\]
		
		Since \(p^n=N_1\pm k \) for some constant \( k \) and $\lim\limits_{N_1 \rightarrow \infty}\frac{1}{p_1}=0$, we have
		
		\[
		\lim_{L \to \infty} \frac{1}{K} = \lim_{N_1(N_1\pm k - c) \to \infty} \frac{1}{\min\{p_1-1, N_1\pm k\}} = 0.
		\]
		By Theorem~\ref{T4}, it follows that the DRCS set \( \mathcal{C} \) in Corollary \ref{CT2} i) is asymptotically optimal.
		
		For the DRCS sets of Corollary~\ref{CT2}~ii) and iii), the proof of asymptotic optimality is similar to that of the DRCS sets in Corollary~\ref{CT2}~i); hence omitted.
		
		ii) By Corollary~\ref{CT2} IV), we have \( K = \min\{p^n,p_1^{n_1}\} \), \( L = (p^{n} - 1)(p_1^{n_1} - c) \), and \( N = p^{n} p_1^{n_1} \). When \(K \geq 4+c \), it holds that
		\[
		K= \min\{p^n,p_1^{n_1}\} > \frac{3 p^{n} p_1^{n_1}}{(p^{n} - 1)(p_1^{n_1} - c)} = \frac{3N}{L}.
		\]
		
		It is obvious that
		\[
		\lim_{L \to \infty} \frac{L}{N} = \lim_{(p^{n} - 1)(p_1^{n_1} - c) \to \infty} \frac{(p^{n} - 1)(p_1^{n_1} - c)}{p^{n} p_1^{n_1}} = 1.
		\]
		
		Since \(p_1=p\pm k \) for constant \( k \), we have
		\[
		\lim_{L \to \infty} \frac{1}{K} = \lim_{(p^{n} - 1)((p \pm k)^{n_1} - c) \to \infty} \frac{1}{\min\{p^n,(p \pm k)^{n_1}\}} = 0.
		\]
		
		By Theorem~\ref{T4}, it follows that the DRCS set \( \mathcal{C} \) of Corollary \ref{CT2} IV) is asymptotically optimal.
		
		For the DRCS sets of Corollary~\ref{CT2}~V), the proof of asymptotic optimality is similar to that of the DRCS sets in Corollary~\ref{CT2}~IV); hence omitted.
	\end{proof}
	
	
	\begin{remark}
		According to the twin prime conjecture, we can select a subsequence consisting of a pair of twin prime numbers, which states that there are infinitely many primes \( p \) such that \( p + 2 \) is also prime \cite{zhang2014bounded}. Polignac proposed a more general conjecture that for every natural number \( r \), there are infinitely many primes \( p \) such that \( p + 2r \) is also a prime \cite{de1851recherches}. The case $r = 1$ of Polignac’s conjecture is the twin prime conjecture. In 2013, Zhang made great contributions in this regard, proving that there are infinitely many pairs of prime numbers that differ
		by 70 million or less \cite{zhang2014bounded}. Maynard reduced the
		difference from 70 million to 600\cite{maynard2015small}. Later, Tao reduced the
		difference from 70 million to 246 in his work on the Polymath
		Project \cite{polymath2014variants}. 
		Thus, infinitely many asymptotically optimal aperiodic DRCS sets can be obtained based on the twin prime conjecture.
	\end{remark}

	\section{Comparison with the Previous Works}
	In this work, we redefined the definition of quasi-Florentine rectangles in a general way. 
	Compared with the classical circular Florentine rectangle, Florentine rectangle, and quasi-Florentine rectangle, the generalized quasi-Florentine rectangle can achieve a significantly larger number of rows for a given $N$ when $N$ is not of the form $p$, $p-1$, $p^n$, or $p^n+1$, where $p$ is a prime. For instance, when $N = 160$, we have $\widetilde{F}(N)=1$ and $F(N) = F_{Q,1}(N) = 4$, whereas $F_{Q,25}(N) = 9$. Table~\ref{table_2} summarizes the values of $\widetilde{F}(N)$, $F(N)$, $F_{Q,1}(N)$, and $F_{Q,N-n}(N)$. It is evident that when $N$ does not take the form $p$, $p-1$, $p^n$, or $p^n+1$, the value of $F_{Q,N-n}(N)$ is strictly larger than both $F(N)$ and $F_{Q,1}(N)$.
	
	The newly introduced generalized quasi-Florentine rectangle structure can be employed to construct asymptotically optimal DRCS sets with more flexible parameters, which are not covered by the results in \cite{shen2024} and \cite{wang2025doppler}. 
	For example, the asymptotically optimal sets shown in Tables~\ref{2} and \ref{3} can only be obtained through our proposed construction to date. 
	In addition, our construction significantly improves the achievable set size 
	compared with the DRCS sets reported in \cite{shen2024} and \cite{wang2025doppler}. 
	To illustrate this improvement, in Table \ref{SA}, we give a few specific examples of the parameters of DRCS set over small alphabets constructed using quasi-Florentine rectangles and Butson-type Hadamard matrices, where the Butson-type Hadamard matrix is obtained through Table \ref{BH} and Lemma \ref{BH_lemma}.
    Tables~\ref{2} and~\ref{3} presented several specific parameters 
	for cases when $N$ is not of the form $p$, $p-1$, $p^n$, or $p^n+1$, 
	where $K_{prev}$ and $K_{{prev}_1}$ denote the maximum set sizes 
	of the DRCS reported in \cite{shen2024} and \cite{wang2025doppler}, respectively, and $\hat{\rho}_{{prev}_1}$ denotes the optimality factor of the $(K_{{prev}_1},N,N-1,N,\Pi)$-DRCS reported in \cite{wang2025doppler}, $\Pi=(-N+1,N-1)\times (-N+1,N-1)$. 
	In these cases, we observe that the set sizes of the DRCSs constructed from the generalized quasi-Florentine rectangles are consistently larger than those of the DRCSs reported in \cite{shen2024} and \cite{wang2025doppler} over the same alphabet. 
    
	\begin{table}[htp]
		\caption{Few Specific Examples of DRCS Set over Small Alphabets.}  
		\label{SA}
		\renewcommand{\arraystretch}{1.3} 
		\resizebox{\textwidth}{!}{
			\begin{tabular}{|c|c|c|c|c|c|c|c|c|c|c|}
				\hline \begin{tabular}{c} 
					Alphabet \\
				\end{tabular} & $K_{{prev}_1}$& $K$ &$N$&$L$& $Z_x$& $Z_y$& $\hat{\theta}_{\max}$& $\hat{\rho}$ &$\hat{\rho}_{{prev}_1}$&\begin{tabular}{c} 
					Butson-type \\
					Hadamard matrix
				\end{tabular} \\
				\hline \multirow{2}{*}{$\mathbb{Z}_3$} & $4$& $6$ &$63$ & $56$& $56$&$56$& $63$ &1.5000 & 1.5591 & $B H(63,3)$  \\
				\cline { 2 - 11 } &$4$ & $16$&$144$&$120$&$120$& $120$&$144$ &1.3248  & 1.5473&  $B H(144,3)$  \\
				\hline \multirow{2}{*}{$\mathbb{Z}_4$} & $4$&$6$&$56$&$49$&$49$&$49$ &$56$ &1.5179 &1.5618 & $B H(56,4)$  \\
				\cline { 2 - 11 } &$4$ &$64$&$5184$&$5040$&$5040$&$5040$ &$5184$&1.0977& 1.5384 & $B H(5184,4)$  \\
				\hline \multirow{2}{*}{$\mathbb{Z}_5$} &$4$ &$8$&$1000$&$868$&$868$& $868$&$1000$& 1.4320 & 1.5395&  $B H(1000,5)$ \\
				\cline { 2 - 11 } & $4$&$16$&$10000$&$9360$&$9360$&$9360$&$10000$&1.2340 &1.5383 &  $B H(10000,5)$  \\
				\hline \multirow{2}{*}{$\mathbb{Z}_6$} & $6$&$49$&$1323$&$1248$&$1248$ &$1248$&$1323$& 1.1300 & 1.3762 &  $B H(1323,6)$  \\
				\cline { 2 - 11 } &$4$ & $128$&$21632$&$21336$&$21336$&$21336$&$21632$&  1.0630 & 1.5382 & $B H(21632,6)$  \\
				\hline 
		\end{tabular}}
	\end{table}	
	\begin{table}[htp]
		\caption{Comparison the set size of DRCS, construction through Corollary \ref{CT2}, where $N$ is not the form $p$, $p-1$, $p^n$, $p^n+1$}
		\label{2}
		\renewcommand{\arraystretch}{1.3} 
		\begin{center}
			\begin{tabular}{|c|c|c|c|c|c|c|c|c|c|c|}
				\hline
				$K_{prev}$ & $K_{prev_1}$ &$K$& $N$ & $L$ &$Z_x$ & $Z_y$ & $\hat{\theta}_{\max}$    & Alphabet & $\hat{\rho}$ &$\hat{\rho}_{{prev}_1}$ \\
				\hline  2&     4&6 &63 &  56 &  56&  56& 63  &$\mathbb{Z}_{7*3^2}$ &  1.5000 &1.5591\\
				\hline   2&   4 &9 &99 & 88 & 88 &88 &  99&$\mathbb{Z}_{11*3^2}$ & 1.3788&1.5514\\
				\hline   1&  4 &12 & 208 & 195 & 195 &195& 208 & $\mathbb{Z}_{13*2^4}$&  1.2754  &1.5444\\
				\hline   1&  4 &16 & 304 & 285 & 285& 285& 304 & $\mathbb{Z}_{19*2^4}$&1.2328& 1.5425\\
				\hline   2& 4 &22 & 759 & 736 & 736& 736& 759  &$\mathbb{Z}_{23*(2^5+1)}$&  1.1726&1.5399\\
				\hline   4& 4 &25 & 925 & 888 & 888 &888& 925  &$\mathbb{Z}_{37*5^2}$& 1.1674 &1.5396\\
				\hline  2& 4 & 30 & 1023 & 992 & 992& 992& 1023 &$\mathbb{Z}_{31*(2^5+1)}$& 1.1455 &1.5395\\
				\hline   6& 6& 46 & 2303 & 2256 & 2256 &2256& 2303  &$\mathbb{Z}_{47*7^2}$& 1.1104& 1.3759\\
				\hline   4& 4 & 60 &3965 & 3904 & 3904 &3904& 3965  &$\mathbb{Z}_{61*(2^6+1)}$&   1.0932& 1.5385 \\
				\hline   2& 4 & 66 &5494 & 5427 & 5427 &5427& 5494  &$\mathbb{Z}_{67*(3^4+1)}$& 1.0869& 1.5384\\
				\hline
			\end{tabular}
		\end{center}
	\end{table}	
	
	\begin{table}[htp]
		\caption{Comparison the set size of DRCS, construction through Corollary \ref{CT2}, where $N$ is not the form $p$, $p-1$, $p^n$, $p^n+1$}
		\label{3}
		\renewcommand{\arraystretch}{1.3} 
		\begin{center}
			\begin{tabular}{|c|c|c|c|c|c|c|c|c|c|c|}
				\hline
				$K_{prev}$ & $K_{prev_1}$ &$K$& $N$ & $L$ &$Z_x$ & $Z_y$ & $\hat{\theta}_{\max}$    & Alphabet & $\hat{\rho}$ &$\hat{\rho}_{{prev}_1}$\\
				\hline 1& 4&9 &160 &  135 &  135&  135& 160  &$\mathbb{Z}_{2^4*(3^2+1)}$ &   1.4283 &1.5463\\
				\hline 2& 4 &25 & 675 & 624 & 624 &624 &  675&$\mathbb{Z}_{3^3*5^2}$ & 1.1932 &1.5401\\
				\hline 1& 4 &25 & 832 & 775 & 775 &775& 832 & $\mathbb{Z}_{2^5*(5^2+1)}$&  1.1880&1.5397 \\
				\hline 1& 4 &25 & 1274 & 1200 & 1200& 1200&1274 & $\mathbb{Z}_{7^2*(5^2+1)}$& 1.1803& 1.5392\\
				\hline 1&  4 &27 & 1350 & 1274 & 1274& 1274& 1350  &$\mathbb{Z}_{(7^2+1)*3^3}$& 1.1722&1.5391\\
				\hline 2&     4 &27 & 1755 & 1664 & 1664 &1664& 1755  &$\mathbb{Z}_{3^3*(2^6+1)}$& 1.1690& 1.5389\\
				\hline 2&    4 & 49 & 3969 & 3840 & 3840& 3840& 3969 &$\mathbb{Z}_{7^2*3^4}$& 1.1144& 1.5385\\
				\hline   1&  4& 64 & 5248 & 5103 & 5103 &5103& 5248  &$\mathbb{Z}_{2^6*(3^4+1)}$& 1.0976& 1.5384\\
				\hline  2&  4 & 81 &9801 & 9600 & 9600 &9600& 9801  &$\mathbb{Z}_{3^4*11^2}$&  1.0831& 1.5383 \\
				\hline   1& 4 & 121 &15246 & 15000 & 15000 &15000&15246  &$\mathbb{Z}_{11^2*(5^3+1)}$& 1.0662&1.5383 \\
				\hline
			\end{tabular}
		\end{center}
	\end{table}	 
	
	\section{Concluding Remarks}
	In this paper, we generalized the existing definition of quasi-Florentine rectangles, establishing them as a powerful combinatorial framework for constructing aperiodic DRCS sets. Based on this generalized quasi-Florentine rectangles, we developed aperiodic DRCS sets whose set sizes are substantially larger than those achievable using the Florentine or quasi-Florentine rectangles defined in \cite{Avik2024}. Moreover, the constructed DRCS sets are asymptotically optimal with respect to the known lower bounds for aperiodic DRCS sets. These findings not only demonstrate the effectiveness of the generalized quasi-Florentine rectangles in DRCS design but also highlight their potential for broader combinatorial and engineering applications. Exploring further applications and theoretical properties of these generalized structures remains an appealing avenue for future research.
	
	\bibliographystyle{ieeetr}
	\bibliography{DRCS_QFR.bib}

\end{document}